\documentclass[journal,twocolumn,11pt,final]{IEEEtran}
\usepackage[utf8]{inputenc}
\usepackage[TS1,T1]{fontenc}
\usepackage{fourier, heuristica}
\usepackage{array, booktabs}
\usepackage{graphicx}
\usepackage[x11names]{xcolor}
\usepackage{colortbl}
\usepackage{caption}
\DeclareCaptionFont{blue}{\color{LightSteelBlue3}}
\usepackage[T1]{fontenc}
\usepackage{float}
\usepackage{amsthm}
\usepackage{amsmath}
\usepackage{amssymb}
\usepackage{dsfont}
\usepackage{graphicx}
\usepackage{tcolorbox}
\usepackage{algpseudocode}
\usepackage{algorithm}
\usepackage{enumitem}   
\newcommand{\subparagraph}{} 
\usepackage[compact]{titlesec}  

\newcommand{\R}{\mathbb{R}}

\newcommand{\N}{\mathbb{N}}

\makeatletter

\DeclareMathOperator*{\argmin}{arg\,min}    

\theoremstyle{plain}
\newtheorem{thm}{\protect\theoremname}

\newtheorem{lem}[thm]{Lemma}
\theoremstyle{definition}
\newtheorem{defn}[thm]{\protect\definitionname}

\newtheorem{remark}{Remark}
\theoremstyle{plain}
\theoremstyle{plain}

\theoremstyle{definition}

\providecommand{\corollaryname}{Corollary}
\providecommand{\propname}{Proposition}
\providecommand{\definitionname}{Definition}
\providecommand{\theoremname}{Theorem}
\providecommand{\assumptionname}{Assumption}

\makeatother

\providecommand{\corollaryname}{Corollary}
\providecommand{\definitionname}{Definition}
\providecommand{\theoremname}{Theorem}

\newcommand{\E}{\mathbb{E}}

\renewcommand{\N}{\mathbb{N}}

\renewcommand{\R}{\mathbb{R}}
\renewcommand{\S}{\mathbb{S}}

\newcommand{\Ac}{\mathcal{A}}

\newcommand{\Cc}{\mathcal{C}}

\newcommand{\Nc}{\mathcal{N}}

\newcommand{\Pc}{\mathcal{P}}
\newcommand{\Qc}{\mathcal{Q}}

\def\algbackskip{\hskip\dimexpr-\algorithmicindent+\labelsep}
\def\LState{\State \algbackskip}

\begin{document}

\title{New Algorithms and Improved Guarantees for One-Bit Compressed Sensing on Manifolds}
\author{
Mark A.~Iwen\(^\dagger\), Eric Lybrand\(^\ddagger\), Aaron A.~Nelson\(^\ddagger\)*, Rayan Saab\(^\ddagger\)
    \thanks{\(^\dagger\)Mathematics and CMSE Departments at Michigan State University. 
   \(^\ddagger\)Mathematics Department at University of California, San Diego.
    *The views expressed in this article are those of the authors and do not reflect the official policy or position of the U.S.~Air Force, Department of Defense, or U.S.~Government.
    }}

\maketitle

\noindent\begin{abstract}
    We study the problem of approximately recovering signals on a manifold from one-bit linear measurements drawn from either a Gaussian ensemble, partial circulant ensemble, or bounded orthonormal ensemble and quantized using \(\Sigma\Delta\) or distributed noise shaping schemes.
    We assume we are given a Geometric Multi-Resolution Analysis, which approximates the manifold, and we propose a convex optimization algorithm for signal recovery.
    We prove an upper bound on the recovery error which outperforms prior works that use memoryless scalar quantization, requires a simpler analysis, and extends the class of measurements beyond Gaussians.
    Finally, we illustrate our results with numerical experiments.
\end{abstract}
\noindent

\section{Introduction}
Compressed sensing \cite{candes2006stable, donoho2006compressed} demonstrates that structured high dimensional signals such as sparse vectors or low-rank matrices can be recovered from few random linear measurements.
Recovery is typically formulated as a convex optimization problem whose minimizer cannot be expressed analytically and must be solved for using numerical algorithms running on digital devices.
Thus, it is necessary to consider the effect of quantization in the design of the recovery algorithms.
Indeed, sparse vector recovery and low-rank matrix recovery have been studied in the presence of various quantization schemes \cite{gunturk2010sigma,huynh2018fast,jacques2013robust,lybrand2017quantization,plan2013one}.
We look to extend these results to account for those structured signals that lie on a compact, low-dimensional submanifold of \(\R^N\) for which we have a Geometric Multi-Resolution Analysis (GMRA) \cite{allard2012multi}.
Our work is motivated  by  the results of Iwen et al.~in \cite{iwen2018recovery} where they assume memoryless scalar quantized Gaussian measurements, and we provide better error bounds that hold for a wider class of measurement ensembles.

As in \cite{iwen2018recovery}, a key component of our technique is the GMRA which approximates the manifold at various levels of refinement.
At each level the GMRA is a collection of approximate tangent spaces about certain known "centers", and the quality of the approximation improves with every level.
Unlike in \cite{iwen2018recovery}, the quantization schemes that we use are  \(\Sigma\Delta\) or distributed noise shaping methods (see, e.g., \cite{gunturk2010sigma,huynh2018fast}) and the compressed sensing measurements that our results apply to include those drawn from Gaussian ensembles, partial circulant ensembles (PCE) or bounded orthonormal ensembles (BOE) (see \cite{huynh2018fast} for precise definitions).
Our proposed reconstruction method is summarized in Algorithm \ref{alg: main}.
This simple algorithm first finds a GMRA center that quantizes to a bit sequence close to the quantized measurements, where "closeness" is determined using a pseudo-metric that respects the quantization; it then optimizes over all points in the associated approximate tangent space to enforce, as much as possible, the consistency of the quantization. 
Using the results of \cite{huynh2018fast} we prove that the quantization error associated with our proposed reconstruction algorithm decays polynomially or exponentially as a function of the number of measurements, depending on the quantization scheme.
This greatly improves on the sub-linear error decay associated with scalar quantization in \cite{iwen2018recovery}.
%
%

\section{Background}
In \cite{iwen2018recovery}, Iwen et al.~study the case where measurements \(y = Ax\) of a signal \(x\) on a manifold \(K\subset \S^{N-1}\) are quantized via memoryless scalar quantization (MSQ).
For a discrete set \(\Ac\) and \(\Qc(x) := \argmin_{z\in \Ac}|x - z|\), the measurements are 
\begin{align*}
    q_j = \Qc(\langle a_j, x \rangle), \quad j = 1, \hdots, m.
\end{align*}
For example, one could take \(\Ac = \{\pm 1\}\) and \(\Qc(\cdot) = \mathrm{sign}(\cdot)\). 
\cite{iwen2018recovery} proposes an algorithm for recovering $x$ from such measurements and shows that the associated error decays like \(O(m^{-1/7})\).
Such slow error decay, associated with MSQ, has also been seen in the context of sparse vector recovery in the compressed sensing literature.
Indeed, it is known in that setting that the error under any reconstruction scheme using MSQ measurements cannot decay faster than \(O(m^{-1})\) \cite{goyal1995quantization} (see also \cite{boufounosJKS14}).
So, to acheive better error rates one must use more sophisticated quantization schemes.
For example, in the sparse vector setting noise shaping techniques such as \(\Sigma\Delta\) and distributed noise shaping leverage redundancy of the measurements to ensure error decay like \(O(m^{-r})\) or \(O(\beta^{-cm})\) for some parameters \(r \in \N\), \(\beta > 1\) that depend on the quantization scheme, e.g., \cite{chou2016distributed, saab2018compressed}.
As we will also use these schemes, we now briefly describe them.

Each of the quantization methods mentioned above employs a state variable \(u \in \R^m\) and quantizes measurements in a recursive fashion:
    \(q_j = \Qc\left(f(y_j,\hdots, y_1, u_{j-1}, \hdots, u_1)\right),\)
where \(f\) is some function designed for the quantization scheme.
The state variable is then updated via the state relation
\(Ax - q = Hu,\)
where \(H: \R^m \to \R^m\) is a lower-triangular Toeplitz matrix.
Important for the analysis (and for practical reasons) is that \(H, f\) are chosen so that whenever \(\|x\|_{\infty}\) is bounded, we have \(\|u\|_{\infty} < C\); the upper bound \(C\) is often referred to as the stability constant of the quantization scheme.
For a more detailed explanation of these noise shaping techniques, the interested reader may refer for example to \cite{huynh2018fast}.
For the sake of expositional simplicity, we will only consider the \(\Sigma\Delta\) setting, but our arguments work for distributed noise shaping under very minor adjustments.


\section{Problem Formulation and Notation}
For an integer \(\ell\),  \([\ell] := \{1, \hdots, \ell\}\). We use \(\gtrsim\) and \(\lesssim\) for inequalities that hold up to a constant; subscripts indicate  the constant depends on a specified parameter.
Let \(K \subset B^{N}_2\) be a \(d\)-dimensional submanifold of the unit \(\ell_2\)-ball in \(\R^N\).
We assume that we have a GMRA of \(K\), which we make precise below.
First, for a set \(T \subset \R^N\) and \(\rho > 0\), define
\begin{align*}
    \mathrm{tube}_{\rho}(T) := \left\{x \in \R^N \; \colon \; \left.\inf\right._{y\in T}\|x - y\|_2 \leq \rho\right\}.
\end{align*}
\begin{defn}[\cite{iwen2013approximation}]\label{def: GMRA}
    Let \(J \in \N\) and \(K_0, \hdots, K_J \in \N\).
    A \textit{GMRA} of \(K\) is a collection \(\{(\Cc_j, \Pc_j)\}_{j \in [J]}\) of centers \(\Cc_{j} = \{c_{j,k}\}_{k \in [K_j]}\) and affine projections 
    \begin{align*}
        \Pc_j = \left\{P_{j,k}\colon \R^N \to \R^N \; \colon \; k \in [K_j]\right\}
    \end{align*}
    with the following properties:
    \begin{enumerate}[leftmargin=*,labelsep=1pt]
    \item\textbf{Affine Projections}.
    Every \(P_{j,k}\) is an orthogonal projection onto some \(d\)-dimensional affine space which contains the center \(c_{j,k}\).
    \item \textbf{Dyadic Structure}.
    The number of centers at each level is bounded by \(|\Cc_j| = K_{j} \leq C_C 2^{dj}\) for an absolute constant \(C_C \geq 1\).
    Moreover, there exist \(C_1 > 0\), \(C_2 \in (0,1]\) such that
    \begin{enumerate}[leftmargin=*]
        \item \(K_{j} \leq K_{j+1}\) for all \(j \in [J-1]\),
        \item \(\|c_{j, k_1} - c_{j, k_2}\|_2 > C_1 2^{-j}\) for all \(j \in [J]\), \(k_1 \neq k_2 \in [K_j]\),
        \item For each \(j \in [J]\setminus\{0\}\) there exists a parent function \(p_j\colon [K_j] \to [K_{j-1}]\) with
        \begin{align*}
            \| c_{j,k} - c_{j-1, p_j(k)}\|_2 
            \leq C_2 \min_{k'\in [K_{j-1}]\setminus\{p_j(k)\}} \| c_{j,k} - c_{j-1, k'}\|_2.
        \end{align*}
    \end{enumerate}
    \item \textbf{Multiscale Approximation}.
    The projectors in \(\Pc_j\) approximate \(K\) in the following sense:
    \begin{enumerate}[leftmargin=*,labelsep=1pt]
        \item There exists \(j_0 \in [J-1]\) such that \(c_{j,k} \in \mathrm{tube}_{C_1 2^{-j-2}}(K)\) for all \(j \geq j_0\) and \(k \in [K_j]\).
        \item For each \(j \in [J]\) and \(z \in \R^N\), let
        \begin{align*}
            c_{j, k_j(z)} \in \argmin_{c_{j,k} \in \Cc_j} \|z - c_{j,k}\|_2.
        \end{align*}
        Then for each \(z \in K\) there exist \(C_z, \widetilde{C}_z > 0\) so that 
         \(   \|z - P_{j,k_j(z)}z\|_2 \leq C_z 2^{-2j}\)
        for all \(j \in [J]\) and
        \begin{align}\label{GMRA part 3b}
            \|z - P_{j,k'}z\|_2 \leq \widetilde{C}_z 2^{-j}
        \end{align}
        whenever \(j \in [J]\) and \(k' \in [K_j]\) satisfy
        \begin{align*}
            \| z - c_{j,k'} \|_2 \leq 16 \max\big\{\|z- c_{j,k_j(z)}\|_2,\; C_1 2^{-j-1}\big\}.
        \end{align*}
    \end{enumerate}
    \end{enumerate}
\end{defn}

\noindent Let \(\{(\Cc_j, \Pc_{j,k})\}_j\) be a GMRA of a smooth compact manifold 
\(K\subset (1-\mu)B_2^N\)
for some \(\mu \in (0,1)\), and define the scale-\(j\) GMRA approximation \(\hat{K}_{j} := \{P_{j, k_j(z)}z\; \colon \; \|z\|_2 \leq 1 \} \cap B_2^N\).
We suppose that \(j_0\) is large enough so that \(\sup_{x\in K}\widetilde{C}_x 2^{-j_0} \leq \mu\) to ensure \(\{P_{j, k'} x \; \colon \; x\in K, \; (j,k') \; \text{as in  3.b of Definition \ref{def: GMRA}} \} \subset B_{2}^N\), and further assume that \(\mathrm{tube}_{C_1 2^{-j_0-2}}(K) \subset B^N_2\) which ensures 
\(\mathcal{C}_j \subset \hat{K}_j\) for \(j \geq j_0\). 
The number of measurements required for our theoretical guarantees to hold will depend on two notions of complexity of \(K\) and the GMRA. 
For \(g \sim \Nc(0, I_{N})\), define
\begin{align*}
     w(S) &:= \E \left.\sup\right._{x\in S} \langle g, x \rangle,\qquad
     \mathrm{rad}(S) := \left.\sup\right._{x\in S} \|x\|_2,
 \end{align*}
 and, for \(j \geq j_0\), define
\begin{align*}
    &C_K := \max\big\{C_1,\left.\sup\right._{z\in K} \widetilde{C}_z\big\}, 
    &S := K \cup \hat{K}_{j}.
\end{align*}
Now, let \(\Qc\)  be a stable \(r^{\text{th}}\) order \(\Sigma\Delta\) quantizer with stability constant \(C(r)\) and associated alphabet \(\Ac\).
Let \(x\in K\), \(A\in \R^{m\times N}\) be a standard Gaussian matrix (or a matrix drawn from a PCE/BOE), \(D_{\epsilon} \in \R^{N\times N}\) a diagonal matrix with random signs (independent of \(A\)) along the diagonal, \(\Phi := AD_{\epsilon}\) and
\begin{align*}
    q := \Qc\left(\Phi x\right).    
\end{align*}
\emph{Our goal, given \(q\) and \(\Phi\), is to approximate \(x\) and show that the associated error bounds decay fast as a function of \(m\).}

A useful fact is that the binary embedding provided by \(\Sigma\Delta\) quantization approximately preserves Euclidean distance via a related pseudo-metric on the quantized vectors, defined as follows.
For \(p \leq m\), define \(\lambda := m/p =: r\widetilde{\lambda} - r + 1\) and \(v\in \R^{\lambda}\) be the (row) vector whose \(j^{th}\) entry \(v_j\) is the \(j^{th}\) coefficient of the polynomial of \((1 + z + \hdots + z^{\widetilde{\lambda}})^r\).
Set \(\gamma = \|v\|_1/\|v\|_2\) and define \(\widetilde{V}\colon\R^m \to \R^p\) by
    \(\widetilde{V} = \frac{9}{8\|v\|_2 \sqrt{p}}I_{p} \otimes v,\)
where \(\otimes\) denotes to the Kronecker product.
Then the pseudo-metric is given by \(d_{\widetilde{V}}(x,y) := \|\widetilde{V}(x-y)\|_2\).

\section{Main Results}
We now present our recovery algorithm and its associated error guarantees.

\begin{algorithm}
\caption{Reconstruction Algorithm}\label{alg: main}
    \begin{algorithmic}
        \LState \textbf{Step 1:} Find \(c_{j,k'} \in \argmin_{c_{j,k} \in \Cc_j}
        \|\widetilde{V}(\Qc(\Phi c_{j,k}) - q )\|_2\).
        
        \LState \textbf{Step 2:} If \(\|\widetilde{V}(\Qc(\Phi c_{j,k'}) - q )\|_2=0\), set \(x^{\sharp}=c_{j,k'}\); else
        \begin{equation*}
            \begin{aligned}
                x^{\sharp} = &\argmin_{z \in \R^N}
                \left\|\widetilde{V}\left(\Phi z - q \right)\right\|_2 \text{  s.t. \(z = P_{j,k'}(z), \; \|z\|_2 \leq 1\)}.
            \end{aligned}
        \end{equation*}
    \end{algorithmic}
\end{algorithm}

\begin{thm}\label{thm: main}
    Suppose we have a GMRA of \(K\subset (1-\mu)B^{N}_2\) at level \(j \geq j_0\). Let \(\Qc\) be a stable \(r^{\text{th}}\) order \(\Sigma\Delta\) quantizer with \(r = O(j)\), associated alphabet \(\Ac = \{\pm(1-\alpha)\sqrt{m}\}\) where \(\alpha^{-1} \gtrsim \mathrm{rad}^2(S-S)2^{2(j+1)}\), and
    \begin{align}\label{eq: num msrments}
        m^\frac{r-1}{r-1/2} \gtrsim
        &\,\gamma^2 (1+\nu)^2\log^4(N)j^2 2^{4j}\\
        &\;\;\times\mathrm{rad}(S-S)^4\max\big\{1, w^2(S-S)\mathrm{rad}^{-2}(S-S)\big\}\nonumber.
    \end{align}
    Then with probability exceeding \(1-e^{-\nu}\), for all \(x\in K\), \(x^{\sharp}\) from Algorithm~\ref{alg: main} satisfies
    \begin{align*}
        \|x^{\sharp} - x\|_2 \lesssim_{r} \widetilde{C}_x 2^{-j}.
    \end{align*}
\end{thm}

\begin{proof}
    By the triangle inequality we have
    \begin{align*}
        \|x^{\sharp} - x\|_2 \leq \|x^{\sharp} - P_{j,k'} x\|_2 + \|P_{j,k'} x - x\|_2.
    \end{align*}
    Both terms are bounded by \(\widetilde{C}_x 2^{-j}\), the first by Lemma~\ref{lem: project} and the second by Lemma \ref{lem: indices} and Equation~\eqref{GMRA part 3b} from Definition~\ref{def: GMRA}.
\end{proof}

\begin{remark}
    As Lemma~4.3 of  \cite{iwen2018recovery} shows, \(w(S-S) \lesssim w(K) + \sqrt{d j}\). This is a suitable bound for coarse GMRA scales, i.e. \(j \lesssim \log(N)\). However, for \(j \gtrsim \log(N)\) one can slightly modify the definition of \(S\) and use the bound \(w(S-S) \lesssim (w(K)+1)\log(N)\) as proven in Lemma~4.5 of \cite{iwen2018recovery}, albeit this requires some modifications to the proof of Theorem \ref{thm: main}. Please see Remark~4.15 of \cite{iwen2018recovery} for more details, as we shall leave the latter case for future work.
\end{remark}




\begin{lem}\label{lem: indices}
    Under the  assumptions of Theorem~\ref{thm: main}, with probability exceeding \(1-e^{-\nu}\), for all \(x\in K\), the center \(c_{j,k'}\) chosen in Step 1 of Algorithm~\ref{alg: main} satisfies
    \begin{align}\label{eq: indices lem}
        \|x-c_{j,k'}\|_2 \leq 16 \max\big\{ \|x - c_{j, k_j(x)}\|_2, C_1 \cdot 2^{-j-1}\big\}.
    \end{align}
\end{lem}
\begin{proof}
    Theorem~5.2 and Remarks~3,~5 of \cite{huynh2018fast} state that if \(m^{\frac{r-1}{r-1/2}} \geq p \gtrsim \gamma^2(1+\nu)^2\log^4(N)\frac{\max\{1, w^2(S-S)\mathrm{rad}^{-2}(S-S)\}}{\alpha^2}\) and we choose \(r = \lfloor \frac{\lambda}{2Ce} \rfloor^{1/2}\), where \(\lambda = m^{\frac{r-1}{r-1/2}}/p\) and \(C > 0\), then with probability exceeding \(1-e^{-\nu}\)
    \begin{align*}
        &\left| d_{\widetilde{V}}\left(\Qc(\Phi c_{j,k}),q\right) - \|c_{j,k} - x\|_2\right| \\ &\quad\lesssim \max\big\{\sqrt{\alpha},\alpha\big\}\mathrm{rad}(S-S) + e^{-c \sqrt{\lambda}}
    \end{align*}
    for all \(c_{j,k} \in \Cc_{j}\). Conditioning on this, we have
    \begin{align*}
         &d_{\widetilde{V}}\left(\Qc(\Phi c_{j, k'}), q\right) \leq d_{\widetilde{V}}\big(\Qc(\Phi c_{j, k_j(x)}), q\big) \\&
        \implies\|c_{j,k'} - x\|_2 \lesssim_{r} \|c_{j,k_j(x)} - x\|_2 \\  & \qquad\qquad\qquad\qquad + \max\big\{\sqrt{\alpha}, \alpha\big\}\mathrm{rad}(S-S) + e^{-c\sqrt{\lambda}}.
    \end{align*}
    For \(c_{j,k'}\) to satisfy~\eqref{eq: indices lem} (i.e., part 3.b of Definition~\ref{def: GMRA}), it suffices to choose \(\alpha\) small and \(\lambda\) large, so that \(
        \max\{\sqrt{\alpha}, \alpha\}\mathrm{rad}(S-S) + e^{-c \sqrt{\lambda}} \leq 15 C_1 2^{-j-1}.\)
    Choosing \(\lambda \gtrsim j^2\) (hence, \(r = O(j)\)) and  \(\alpha^{-2} \gtrsim \mathrm{rad}^4(S-S)2^{4(j+1)}\) realizes the above bound. 
\end{proof}

\begin{lem}\label{lem: project}
    Under the assumptions of Theorem~\ref{thm: main}, with probability exceeding \(1-e^{-\nu}\), for all \(x\in K\), \(x^{\sharp}\) from Algorithm~\ref{alg: main} satisfies \(\|x^{\sharp} - P_{j,k'}x\|_2 \lesssim_{r} \widetilde{C}_x 2^{-j}\).
\end{lem}
\begin{proof}
    By optimality of \(x^{\sharp}\) and feasibility of \(P_{j,k'}x\), the triangle inequality gives us
    \begin{align*}
        0 &\leq \left\|\widetilde{V}\left(\Phi P_{j,k'}x - q \right)\right\|_2 - \Big\|\widetilde{V}\big(\Phi x^{\sharp} - q \big)\Big\|_2\\
        &\leq 2\left\|\widetilde{V}\left(\Phi P_{j,k'}x - q \right)\right\|_2 - \Big\|\widetilde{V}\Phi \big(x^{\sharp} - P_{j,k'}x \big)\Big\|_2.
    \end{align*}
    Define \(q^* := \Qc\left(\Phi P_{j,k'}x\right)\).
    Then
    \begin{align*}
        \Big\|\widetilde{V}\Phi \big(x^{\sharp} - P_{j,k'}x \big)\Big\|_2 
        &\leq 2\left\|\widetilde{V}\left(\Phi P_{j,k'}x - q^* + q^* - q \right)\right\|_2\\
        &\leq 2 \left\|\widetilde{V}\left(\Phi P_{j,k'}x - q^*\right)\right\|_2 + 2d_{\widetilde{V}}(q, q^*).
    \end{align*}
    Lemma~4.5 of \cite{huynh2018fast} states that \(\|\widetilde{V}(\Phi P_{j,k'}x - q^*)\|_2 \lesssim_{r} e^{-c\sqrt{\lambda}} \lesssim_{r} 2^{-j}\), while Theorem~5.2 of \cite{huynh2018fast} and~\eqref{GMRA part 3b} (via Lemma~\ref{lem: indices}) imply \(d_{\widetilde{V}}(q, q^*) \lesssim \widetilde{C}_x 2^{-j}\) with probability exceeding \(1-e^{-\nu}\).
    Therefore, we have
    \begin{align}\label{x_sharp-Px small}
         \Big\|\widetilde{V}\Phi \big(x^{\sharp} - P_{j,k'}x \big)\Big\|_2 &\lesssim_{r}  \widetilde{C}_x 2^{-j}.
    \end{align}
    By the definition of \(S\), we have \(x^{\sharp}, P_{j,k'}x \in S\).
    Equation~(5.8) in the proof of Theorem~5.2 of \cite{huynh2018fast} states that with probability exceeding \(1-e^{-\nu}\)
         \(\Big\|\widetilde{V}\Phi \big(x^{\sharp} - P_{j,k'}x \big)\Big\|_2
         \gtrsim \big\|x^{\sharp} - P_{j,k'}x\big\|_2 - 2^{-j}.\)
    With~\eqref{x_sharp-Px small}, this yields
    \(\|x^{\sharp} - P_{j,k'}x\|_2 \lesssim_r  \widetilde{C}_x 2^{-j}\).
\end{proof}

\begin{remark}\label{rem:err_poly}
If one does not choose \(r\) and \(\alpha\) as in the proofs above, but works with \(r\) large enough and \(\alpha\) small enough,  a version of Theorem~\ref{thm: main} with
 $   \|x^\sharp-x\|_2
    \lesssim \widetilde{C}_x 2^{-j} + \max\{\sqrt{\alpha},\alpha\}\mathrm{rad}(S-S) + (m/p)^{-r+1/2}$
holds for $m$, \(p\) as in the proof of Lemma~\ref{lem: indices}.
\end{remark}

\begin{remark}
    When \(\Sigma\Delta\) quantization is replaced by distributed noise shaping, Theorem~\ref{thm: main}  holds with  \(j^2 2^{4j}\)  in the lower bound \eqref{eq: num msrments}  replaced by \(j 2^{4j}\), and  Remark~\ref{rem:err_poly} holds with \(\beta^{-m/p}\) replacing \((m/p)^{-r+1/2}\).
\end{remark}

\section{Numerical Simulations}
To simulate Algorithm~\ref{alg: main}, we take \(K=\S^2\) embedded in \(\R^{20}\) and construct a GMRA up to level \(j_{\max}=15\) using 20,000 data points sampled uniformly from \(K\).
We randomly select a test set of 100 points \(x\in K\) for use throughout all experiments.
In each experiment (i.e., point in Figure~\ref{fig: error}), compressed sensing measurements \(y=\Phi x=m^{-1/2}AD_{\epsilon}x\) are taken for each test point, with \(A\sim\Nc(0,I_{m\times N})\) and \(D_{\epsilon}\) a diagonal \(N\times N\) matrix of random \(\pm 1\)s.
We recover $x^\sharp$ from the \(r^{\text{th}}\) order \(\Sigma\Delta\) measurements \(\Qc(y)\) via Algorithm~\ref{alg: main} where, for practical reasons, the alphabet from Theorem~\ref{thm: main} is modified to be \(\mathcal{A}=\{\pm1\}\).
We vary \(\lambda=m/p\) for fixed \(r\), \(p\), and refinement scale \(j\).
As in Remark~\ref{rem:err_poly}, the reconstruction error decays as a function of \(\lambda\) until reaching a floor due to the refinement level of the GMRA.


\begin{figure}[h]
    \centering
    \includegraphics[width=\linewidth]{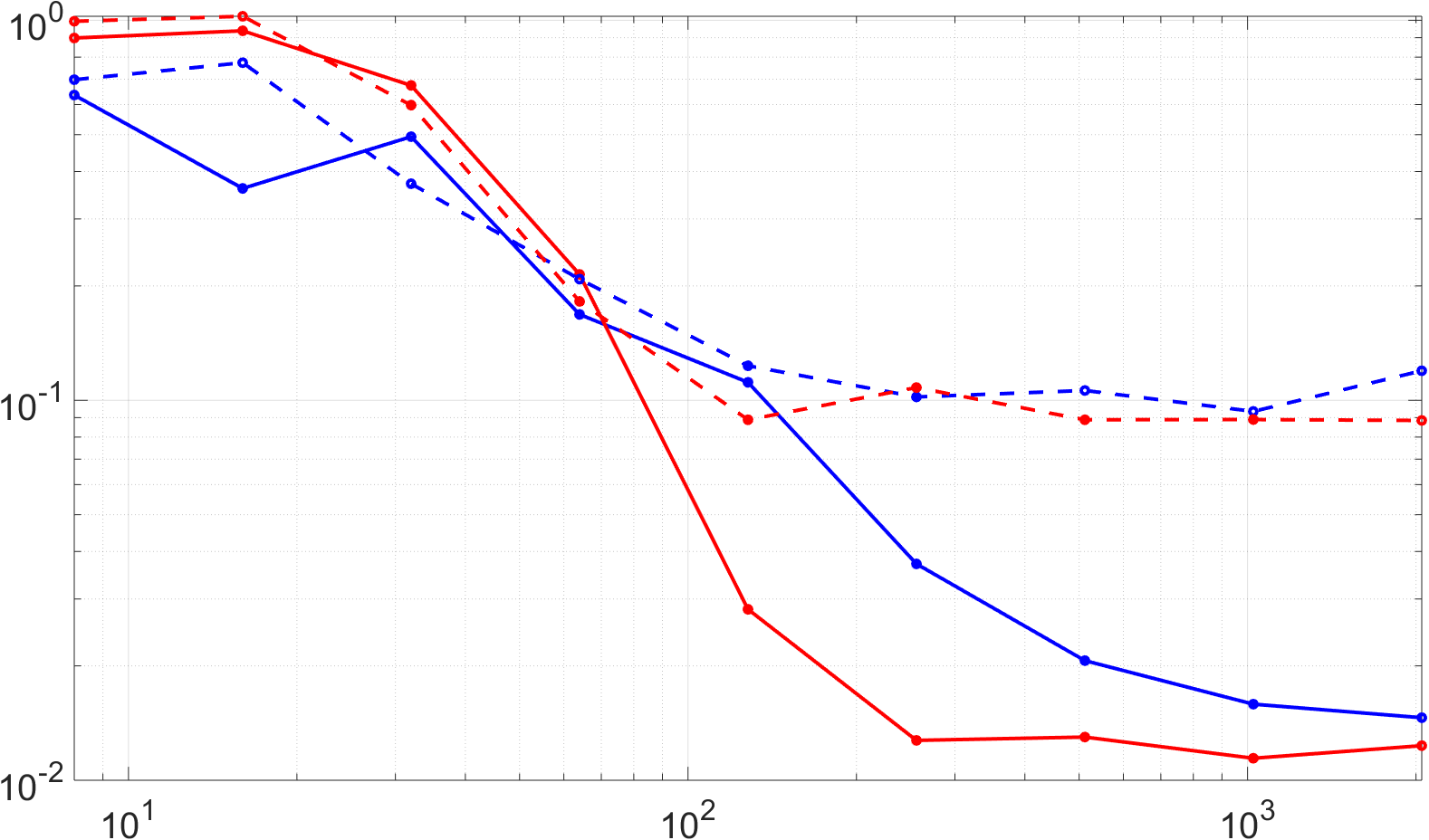}
    \caption{Log-scale plot of average relative reconstruction error from Algorithm~\ref{alg: main} as a function of \(\lambda=m/p\) for \(p=10\).
    Solid lines correspond to GMRA refinement level \(j=12\); dashed lines to  \(j=6\).
    Blue and red plots represent \(r=2,4\) (resp.).
    For each \(j\), reconstruction error decays as a function of \(\lambda\) until reaching a floor due to error in the GMRA approximation of \(K\).
    }
    \label{fig: error}
\end{figure}

\section{Acknowledgements}
We thank F.~Krahmer, S.~Krause-Solberg, and J.~Maly for  sharing their GMRA code, which they adapted from that provided by M.~Maggioni.
M.~A.~Iwen was supported in part by NSF CCF-1615489;
R.~Saab was supported in part by NSF DMS-1517204.
\bibliographystyle{plain}
\bibliography{citations}

\begin{thebibliography}{10}

\bibitem{allard2012multi}
William~K Allard, Guangliang Chen, and Mauro Maggioni.
\newblock Multi-scale geometric methods for data sets ii: Geometric
  multi-resolution analysis.
\newblock {\em Applied and Computational Harmonic Analysis}, 32(3):435--462,
  2012.

\bibitem{boufounosJKS14}
Petros~T. Boufounos, Laurent Jacques, Felix Krahmer, and Rayan Saab.
\newblock Quantization and compressive sensing.
\newblock {\em preprint arXiv:1405.1194}, 2014.

\bibitem{candes2006stable}
Emmanuel~J. Candes, Justin~K. Romberg, and Terence Tao.
\newblock Stable signal recovery from incomplete and inaccurate measurements.
\newblock {\em Comm. Pure Appl. Math.}, 59(8):1207--1223, 2006.

\bibitem{chou2016distributed}
Evan Chou and C.~Sinan G{\"u}nt{\"u}rk.
\newblock Distributed noise-shaping quantization: I. beta duals of finite
  frames and near-optimal quantization of random measurements.
\newblock {\em Constr. Approx.}, 44(1):1--22, 2016.

\bibitem{donoho2006compressed}
David~L. Donoho.
\newblock Compressed sensing.
\newblock {\em IEEE Trans. Inf. Theory}, 52(4):1289--1306, 2006.

\bibitem{goyal1995quantization}
Vivek~K. Goyal, Martin Vetterli, and Nguyen~T. Thao.
\newblock Quantization of overcomplete expansions.
\newblock In {\em Data Compression Conference, 1995. DCC'95. Proceedings},
  pages 13--22. IEEE, 1995.

\bibitem{gunturk2010sigma}
C.~Sinan G{\"u}nt{\"u}rk, Mark Lammers, Alex Powell, Rayan Saab, and
  {\"O}zg{\"u}r Yilmaz.
\newblock Sigma delta quantization for compressed sensing.
\newblock In {\em Information Sciences and Systems (CISS), 2010 44th Annual
  Conference on}, pages 1--6. IEEE, 2010.

\bibitem{huynh2018fast}
Thang Huynh and Rayan Saab.
\newblock Fast binary embeddings, and quantized compressed sensing with
  structured matrices.
\newblock {\em preprint arXiv:1801.08639}, 2018.

\bibitem{iwen2018recovery}
Mark~A. Iwen, Felix Krahmer, Sara Krause-Solberg, and Johannes Maly.
\newblock On recovery guarantees for one-bit compressed sensing on manifolds.
\newblock {\em preprint arXiv:1807.06490}, 2018.

\bibitem{iwen2013approximation}
Mark~A. Iwen and Mauro Maggioni.
\newblock Approximation of points on low-dimensional manifolds via random
  linear projections.
\newblock {\em Inf. Inference}, 2(1):1--31, 2013.

\bibitem{jacques2013robust}
Laurent Jacques, Jason~N. Laska, Petros~T. Boufounos, and Richard~G. Baraniuk.
\newblock Robust 1-bit compressive sensing via binary stable embeddings of
  sparse vectors.
\newblock {\em IEEE Trans. Inf. Theory}, 59(4):2082--2102, 2013.

\bibitem{lybrand2017quantization}
Eric Lybrand and Rayan Saab.
\newblock Quantization for low-rank matrix recovery.
\newblock {\em Inf. Inference}, 2017.

\bibitem{plan2013one}
Yaniv Plan and Roman Vershynin.
\newblock One-bit compressed sensing by linear programming.
\newblock {\em Comm. Pure Appl. Math.}, 66(8):1275--1297, 2013.

\bibitem{saab2018compressed}
Rayan Saab, Rongrong Wang, and {\"O}zg{\"u}r Y{\i}lmaz.
\newblock From compressed sensing to compressed bit-streams: practical
  encoders, tractable decoders.
\newblock {\em IEEE Trans. Inf. Theory}, 64(9):6098--6114, 2018.

\end{thebibliography}

\end{document}